\documentclass[conference]{IEEEtran}
\usepackage{cite}

\ifCLASSINFOpdf
  \usepackage[pdftex]{graphicx}
  \graphicspath{{figures/}}
  \DeclareGraphicsExtensions{.pdf,.jpeg,.png}
\else
  \usepackage[dvips]{graphicx}
  \graphicspath{{figures/}}
  \DeclareGraphicsExtensions{.eps}
\fi

\usepackage{amsmath}
\usepackage{amssymb}
\usepackage{shortsym}

\usepackage{amsthm}

\newtheorem{theorem}{Theorem}
\newtheorem{lemma}{Lemma}
\newtheorem{corollary}{Corollary}

\usepackage{algorithmicx}

\usepackage{array}

\ifCLASSOPTIONcompsoc
  \usepackage[caption=false,font=normalsize,labelfont=sf,textfont=sf]{subfig}
\else
  \usepackage[caption=false,font=footnotesize]{subfig}
\fi

\usepackage{dblfloatfix}
\usepackage{url}

\usepackage{tikz}
\usepackage{pgfplots}
\usetikzlibrary{calc}
\usetikzlibrary{decorations.pathreplacing}
\usetikzlibrary{arrows}
\usetikzlibrary{automata}
\usetikzlibrary{positioning}
\usepgfplotslibrary{ternary}
\usepackage[super]{nth}

\newcommand{\eg}[0]{e.g.}
\newcommand{\ie}[0]{i.e.}
\newcommand{\T}[0]{\ensuremath{\top}}
\newcommand{\BFchat}[0]{{\boldcommand{\hat{\mathsf{c}}}}}
\newcommand{\BFuhat}[0]{{\boldcommand{\hat{\mathsf{u}}}}}

\newcommand{\BFctilde}[0]{{\boldcommand{\tilde{\mathsf{c}}}}}

\newcommand{\BFGtilde}[0]{{\boldcommand{\tilde{\mathsf{G}}}}}

\DeclareMathOperator{\rk}{rank}
\newcommand{\eqdef}[0]{\ensuremath{\triangleq}}
\newcommand{\eqshould}[0]{\ensuremath{\overset{!}{=}}}
\newcommand{\pnR}[0]{\ensuremath{\overline{r}}}
\newcommand{\pR}[0]{\ensuremath{r}}
\newcommand{\Exp}[1]{\ensuremath{\IE\!\left[#1\right]}}
\newcommand{\Prob}[1]{\ensuremath{\operatorname{Pr}\!\left[#1\right]}}
\newcommand{\Var}[1]{\ensuremath{\operatorname{Var}\!\left[#1\right]}}
\newcommand{\Ivec}[0]{\ensuremath{\mathfrak{v}}}
\newcommand{\Imat}[0]{\ensuremath{\mathfrak{m}}}

\usepackage{mathtools}

\DeclarePairedDelimiter\floor{\lfloor}{\rfloor}
\DeclarePairedDelimiter\abs{\lvert}{\rvert}

\newcommand{\dd}{\,\mathrm{d}\hspace{0.07em}}

\usepackage{amsmath}
\allowdisplaybreaks

\usepackage{hyperref}

\begin{document}
\bstctlcite{IEEEexample:BSTcontrol}

\title{Babel Storage: Uncoordinated Content Delivery from Multiple Coded Storage Systems}

\author{\IEEEauthorblockN{Joachim Neu}
\IEEEauthorblockA{Stanford University\\
Stanford, CA, USA\\
Email: jneu@stanford.edu}
\and
\IEEEauthorblockN{Muriel M\'edard}
\IEEEauthorblockA{Massachusetts Institute of Technology\\
Cambridge, MA, USA\\
Email: medard@mit.edu}}

\maketitle

\begin{abstract}
In future content-centric networks, content is identified independently of its location.
From an end-user's perspective, individual storage systems dissolve into a seemingly omnipresent structureless `storage fog'.
Content should be delivered oblivious of the network topology, using multiple storage systems simultaneously, and at minimal coordination overhead.
Prior works have addressed the advantages of error correction coding for distributed storage and content delivery separately.
This work takes a comprehensive approach to highlighting the tradeoff between storage overhead and transmission overhead in uncoordinated content delivery from multiple coded storage systems.

Our contribution is twofold.
First, we characterize the tradeoff between storage and transmission overhead when all participating storage systems employ the same code.
Second, we show that
the resulting
stark inefficiencies can be avoided when storage systems use diverse codes.
What is more,
such code diversity
is not just technically desirable, but
presumably
will be
the reality in the increasingly heterogeneous networks of the future.
To this end,
we show that a mix of Reed-Solomon, low-density parity-check and random linear network codes achieves close-to-optimal performance at minimal coordination and operational overhead.
\end{abstract}

\begin{IEEEkeywords}
Error correction coding,
content-centric networks,
distributed storage,
content delivery.
\end{IEEEkeywords}

\IEEEpeerreviewmaketitle

\section{Introduction}
\label{sec:introduction}

In current communication networks, content is addressed by its (alleged) location, \eg, through its URL.
Content delivery is carried out only by the hosting storage system.
Different `hacks', \eg,
GeoDNS and IPv4 anycast \cite{book_intro_computer_networks},
are used to overcome the shortcomings of this approach.
In future \emph{content-centric networks}, content is identified independent of its location \cite{DBLP:conf/sigcomm/KoponenCCEKSS07}.
From an end-user's perspective, individual storage systems dissolve into a seemingly omnipresent, structureless and thus flexible `storage fog'.
This storage facility delivers the desired data without requiring the user to know about the data's location.
File delivery can be carried out by multiple storage systems simultaneously, and transparently to the user.

Storage systems use error correcting codes \cite{DBLP:books/daglib/0027592} to improve reliability in the presence of temporary and permanent disk failures.
Different code families have been investigated for application in coded storage systems, such as
Reed-Solomon codes (RS-Cs) \cite{DBLP:conf/adl/GoldbergY98,DBLP:conf/asplos/KubiatowiczBCCEGGRWWWZ00,DBLP:conf/hotstorage/RashmiSGKBR13},
random linear network codes (RLN-Cs) \cite{acedanski2005good},
low-density parity-check codes (LDPC-Cs) \cite{pt:03:ldpc,DBLP:journals/tcom/ParkLM18,DBLP:journals/corr/LubyPRMA17,DBLP:journals/corr/Luby16}
and locally repairable codes (LR-Cs) \cite{DBLP:journals/pvldb/SathiamoorthyAPDVCB13,DBLP:journals/tit/PapailiopoulosD14},
among many other codes \cite{DBLP:conf/usenix/HuangSXOCG0Y12,DBLP:journals/pieee/DimakisRWS11,DBLP:journals/ftcit/ShokrollahiL09,DBLP:journals/tos/LiL14,staircasecodes}.
As there is no clear `best' coding system, a diverse mix of codes is expected to be found among the different storage units of a storage fog.
A separate line of works has studied the benefits of introducing redundant information via error correcting codes for content delivery in wireless networks \cite{DBLP:journals/icl/LivaPC10}, such as for caching
\cite{DBLP:journals/tit/Maddah-AliN19,DBLP:conf/isit/ReisizadehMM18},
using device-to-device communications \cite{DBLP:journals/tcom/PiemonteseA19},
or in gossip protocols \cite{DBLP:journals/tit/DebMC06,DBLP:journals/jacm/Haeupler16}.

How can multiple storage systems, employing a mix of different error correcting codes, jointly deliver content to a user with high throughput, low coordination overhead, and low complexity?
\emph{In particular, how can we combine the two so far isolated approaches of coding for storage and coding for content delivery?}
This is the question guiding our work in this paper.
We present our results as follows.
In Section~\ref{sec:system-model}, we introduce the
system model.
We then look at two extreme cases regarding the used codes.
First, in Section~\ref{sec:same-codes}, we use exactly the same code on all participating storage systems.
We show that storage redundancy, while originally introduced to improve reliability, also helps in efficient file delivery.
We characterize the tradeoff between storage and transmission overhead.
Second, in Section~\ref{sec:different-codes}, we examine the opposite extreme of using different codes.
In particular, we combine RS-Cs, RLN-Cs and LDPC-Cs, and demonstrate that the combination behaves similar to a low-field-size random code.
We summarize the implications of our findings as relevant to system implementers and network operators in Section~\ref{sec:conclusion}.

\section{System Model}
\label{sec:system-model}

\begin{figure*}[!t]
    \centering
    \begin{tikzpicture}
        \node[inner sep=0pt] (u_icon) at (0,0) {\includegraphics[width=1cm]{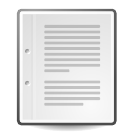}};
            \node (u) at ($(u_icon.east)+(+0.4,-0.0)$) {$\BFu$};

        \node[inner sep=0pt] (s1) at (5,+1) {\includegraphics[width=1cm]{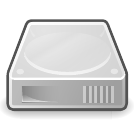}};
            \node (s1_c) at ($(s1.west)+(-0.5,+0.0)$) {$\BFc_1$};
            \node (s1_chat) at ($(s1.east)+(+0.5,+0.0)$) {$\BFchat_1$};
        \node[inner sep=0pt] (sS) at (5,-1) {\includegraphics[width=1cm]{tango-harddisk}};
            \node (sS_c) at ($(sS.west)+(-0.5,+0.0)$) {$\BFc_S$};
            \node (sS_chat) at ($(sS.east)+(+0.5,+0.0)$) {$\BFchat_S$};
        \node[inner sep=0pt] (sdots) at (5,0) {$\vdots$};

        \draw[->] (u.east) -- (s1_c.west) node[midway,fill=white] {$\BFc_1 = \BFu \BFG_1$};
        \draw[->] (u.east) -- (sS_c.west) node[midway,fill=white] {$\BFc_S = \BFu \BFG_S$};

        \node (s1_ctilde) at ($(s1.east)+(+3.5,+0.0)$) {$\BFctilde_1$};
        \node (sS_ctilde) at ($(sS.east)+(+3.5,+0.0)$) {$\BFctilde_S$};
        \draw[->,dashed] (s1_chat) -- (s1_ctilde);
        \draw[->,dashed] (sS_chat) -- (sS_ctilde);
        \draw [decorate,decoration={brace,amplitude=5pt,raise=0pt},yshift=0pt] (s1_ctilde.north east) -- (sS_ctilde.south east);
        \node (collect_ctilde) at ($(sdots)+(+6.6,+0)$) {$\underbrace{\left[ \BFctilde_1 \dots \BFctilde_S \right]}_{\BFctilde} = \BFu \underbrace{\left[ \BFGtilde_1 \dots \BFGtilde_S \right]}_{\BFGtilde}$};

        \node (uhat_icon) at (17,0) {\includegraphics[width=1cm]{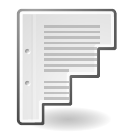}};
        \node (uhat) at ($(uhat_icon.west)+(-0.4,-0.0)$) {$\BFuhat$};
        \draw[->] (collect_ctilde.east) -- (uhat.west);

        \node at ($(u_icon)+(+0.0,+2.0)$) {\sc File};
        \node at ($(sdots)+(-2.7,+2.0)$) {\sc Encoding};
        \node at ($(sdots)+(+0.0,+2.0)$) {\sc Storage};
        \node at ($(sdots)+(+2.5,+2.0)$) {\sc Transmission};
        \node at ($(uhat)+(-1.3,+2.0)$) {\sc Decoding};
        \node at ($(uhat_icon)+(+0.0,+2.0)$) {\sc File};
    \end{tikzpicture}
    \caption{
    A file is stored and distributed using $S$ different storage systems in the following way.
    For each storage system $i \in \left\{1,...,S\right\}$,
    the file $\BFu$ is encoded into $\BFc_i$ which is stored.
    Some symbols of $\BFc_i$ get erased over time;
    the remaining ones are collected in the vector $\BFchat_i$.
    Of these, some are chosen for transmission;
    the symbols that arrive at the requesting user are denoted as $\BFctilde_i$.
    As long as $\rk\left[ \BFGtilde_1 \dots \BFGtilde_S \right] = k$,
    the file $\BFuhat = \BFu$ can be decoded.
    }
    \label{fig:system-model}
\end{figure*}

We denote the finite field of characteristic $p$ as $\IF_p$ and its extension field of degree $m$ as $\IF_{p^m}$.
We sometimes write $\IF_q$ for $\IF_{p^m}$ implying that $q = p^m$.
Note that $\IF_{q^m}$ is isomorphic to the $m$-dimensional $\IF_q$-vector space $\IF_q^m$.
We use lower-case Greek letters for field elements, lower-case Latin bold letters for vectors, and upper-case Latin bold letters for matrices.
As customary in coding theory, vectors denote row vectors unless stated otherwise.
Transposition is $(.)^\T$.

Fix a block length $n$, dimension $k \leq n$, and field $\IF_q$, then a $[n, k]_q$ linear block code $\CC$ with rate $R\eqdef\frac{k}{n}$ is a $k$-dimensional subspace of $\IF_q^n$. The code can equivalently be represented using a generator matrix $\BFG \in \IF_q^{k \times n}$ or a parity-check matrix $\BFH \in \IF_q^{(n-k) \times n}$,
\begin{IEEEeqnarray}{rCl}
    \CC = \left\{ \BFc \;\middle\vert\; \exists \, \BFu \in \IF_q^k : \BFc = \BFu \BFG \right\} = \left\{ \BFc \;\middle\vert\; \BFH \BFc^\T = \BFzero^\T  \right\} \subseteq \IF_q^n.   \IEEEeqnarraynumspace
\end{IEEEeqnarray}
Note that neither $\BFG$ nor $\BFH$ are unique
and that $\BFG$ fixes the mapping between information symbols $\BFu$ and codeword symbols $\BFc$ as
\begin{IEEEeqnarray}{rCl}
    \BFc = \BFu \BFG.
\end{IEEEeqnarray}
During storage and transmission, some number $n_{\text{E}}$ of symbols of the codeword $\BFc$ might get erased, so that only a subset $\BFctilde \in \IF_q^{n-n_{\text{E}}}$ is received.
The $\BFctilde$ relate to $\BFu$ as $\BFctilde = \BFu \BFGtilde$ via a reduced $\BFGtilde \in \IF_q^{k \times (n-n_{\text{E}})}$ where those columns of $\BFG$ are removed that correspond to the erased symbols in $\BFc$. As long as $\BFGtilde$ is full rank (i.e., $\rk\BFGtilde = k$), $\BFctilde = \BFu \BFGtilde$ has a unique solution and $\BFu$ can be recovered from $\BFctilde$, \eg, using Gaussian elimination in runtime complexity $\CO(k^3)$.

The storage and distribution of a file using $S$ different storage systems is depicted in Fig.~\ref{fig:system-model} and proceeds as follows.
For each storage system $i \in \left\{1,...,S\right\}$, the file $\BFu \in \IF_q^k$ is encoded using a $[n_i,k]_q$ code into $\BFc_i = \BFu \BFG_i \in \IF_q^{n_i}$.
These codewords $\BFc_i$ are stored.
Some symbols get erased over time, the remaining ones are denoted $\BFchat_i \in \IF_q^{\hat{n}_i}$.
Of these, some are chosen for transmission, so that $\BFctilde_i \in \IF_q^{\tilde{n}_i}$ arrive at the requesting user.
Let $\BFctilde \eqdef \left[ \BFctilde_1 \dots \BFctilde_S \right]$ and $\BFGtilde \eqdef \left[ \BFGtilde_1 \dots \BFGtilde_S \right]$.
As long as $\rk\BFGtilde = k$, the file $\BFuhat = \BFu$ can be decoded.

The goal is to transmit as few symbols as possible, but enough so that the user can decode as soon as possible with very high probability.
A codeword symbol is useful if and only if its corresponding column increases the rank of the matrix $\BFGtilde$.
There should be no coordination, neither among the storage systems nor by the user, regarding the selection of symbols to be transmitted.
Instead, using the metaphor of a `digital fountain' \cite{DBLP:conf/sigcomm/ByersLMR98}, upon request each storage system sends a stream of symbols until the user asks it to stop.
The system performance is dominated by a \emph{coupon collector's problem} -- how many transmissions (some of which might not add new information) are necessary to collect the required $k$ independent pieces?
In Sections~\ref{sec:same-codes} and \ref{sec:different-codes} we show how storage redundancy and code diversity can mitigate
the loss in efficiency due to inadvertently transmitting duplicate pieces.

\section{Minimum Code Diversity}
\label{sec:same-codes}

In this section, we consider an extreme case of the system depicted in Fig.~\ref{fig:system-model}.
All $S$ storage systems use the same $[n, k]_q$ maximum distance separable (MDS) code, \ie, the information can be decoded from any $k$ codeword symbols (and equivalently, any $k$ columns of the code's $\BFG$ are linearly independent).
The storage systems take turns transmitting a randomly selected codeword symbol to the user, \eg, for $S=2$, the $\nth{1}$, $\nth{3}$, ..., and $\nth{2}$, $\nth{4}$, ..., symbols arrive from the first and second system, respectively.
While each storage system ensures uniqueness among the symbols it sends, the user might receive the same symbol (and thus redundant information) multiple times, from different systems.

\begin{figure}[!t]
    \centering
    \begin{tikzpicture}[->,>=stealth',shorten >=1pt,auto,node distance=1.8cm,semithick]
        \tikzstyle{every state}=[minimum size=1.1cm]

        \node[state] (A)   {$n$};
        \node[state] (C) [right of=A,xshift=+0.9cm]   {$i$};
        \node[state] (D) [right of=C]   {$i-1$};
        \node[state] (F) [right of=D,xshift=+0.9cm]   {$0$};

        \path (A) edge [loop below] node {$\pnR_{n,\ell}$} (A);
        \path (A) edge [draw=none] node [xshift=-0.35cm] {$\pR_{n,\ell}$} (C);
        \path (A) edge [] node [midway,xshift=-0.1cm,yshift=-0.2cm,fill=white] {$\dots$} (C);

        \path (C) edge [loop below] node {$\pnR_{i,\ell}$} (C);
        \path (C) edge [] node {$\pR_{i,\ell}$} (D);
        \path (D) edge [loop below] node {$\pnR_{i-1,\ell}$} (D);

        \path (D) edge [draw=none] node [xshift=-0.25cm] {$\pR_{i-1,\ell}$} (F);
        \path (D) edge [] node [midway,xshift=-0.1cm,yshift=-0.2cm,fill=white] {$\dots$} (F);
        \path (F) edge [loop below] node {$1$} (F);
    \end{tikzpicture}
    \caption{Time-varying Markov chain modeling the evolution of the number of symbols $U_\ell$ that were not seen among the first $\ell$ receive symbols.}
    \label{fig:markov-chain}
\end{figure}
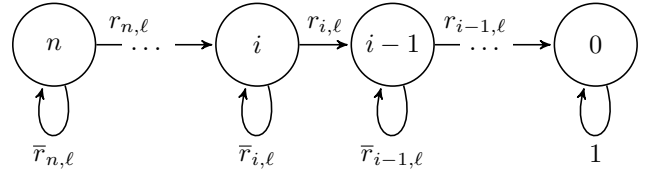

In the following,
we analyze
the underlying coupon collector's problem,
both for large systems (block length $n\to\infty$) as well as for systems with finite $n$.
We model the number of symbols that were not seen among the first $\ell$ received symbols as the random variables $U_\ell$
(`\underline{u}nknown').
The $\left\{ U_\ell \right\}_{\ell=0}^{Sn}$ are described by the time-varying Markov chain depicted in Fig.~\ref{fig:markov-chain}, with states $\CU = \left\{ 0, ..., n \right\}$, initial state $U_0$ distributed as
\begin{IEEEeqnarray}{rCl}
    U_0 &=& \left\{
        \begin{array}{ll}
        n & \text{with probability (w.p.) } 1
        \end{array}\right.
\end{IEEEeqnarray}
and transitions
\begin{IEEEeqnarray}{rCl}
    U_{\ell+1} \mid U_{\ell} &=& \left\{
        \begin{array}{ll}
            U_{\ell} - 1
                & \text{ w.p. } \pR_{U_\ell,\ell}   \\
            U_{\ell}
                & \text{ w.p. } \pnR_{U_\ell,\ell}
        \end{array}\right.
\end{IEEEeqnarray}
where
\begin{IEEEeqnarray}{rCl}
    \pR_{i,\ell} \eqdef \frac{i}{n-\floor*{\frac{\ell}{S}}}
    \qquad{}
    \text{and}
    \qquad{}
    \pnR_{i,\ell} \eqdef 1-\pR_{i,\ell}
\end{IEEEeqnarray}
model coupon collecting and
denote the probability of receiving a new or duplicate symbol in the $\ell$-th transmission, respectively, if at that time $i$ symbols are unknown to the receiver.
In the $\ell$-th round,
the active storage system chooses uniformly at random one of the
$n-\floor*{\frac{\ell}{S}}$
pieces it has not transmitted in previous rounds.
Out of these pieces,
$i$ pieces have not been seen previously by the receiver,
hence the probability $\pR_{i,\ell}$ for receiving novel information.
Using this model, the probability mass functions (PMFs) $P_{U_\ell}(u_\ell)$ can be obtained for all $\ell$
for any fixed $n$ and $S$.

\begin{lemma}
\label{thm:recursive-expected-value}
The following recursion holds:
\begin{IEEEeqnarray}{rCl}
    \Exp{U_{\ell+1}} = \underbrace{\left( 1 - \frac{1}{n-\floor*{\frac{\ell}{S}}} \right)}_{\eqdef \, a_\ell} \Exp{U_\ell}
\end{IEEEeqnarray}
\end{lemma}
\begin{proof}
\begin{IEEEeqnarray}{rCl}
    \IEEEeqnarraymulticol{3}{l}{
        \Exp{ U_{\ell+1} } = \Exp{ \Exp{ U_{\ell+1} | U_{\ell} } }
    }\\
    \quad&=&
        \Exp{ (U_{\ell}-1) \, \pR_{U_\ell,\ell} + U_{\ell} \, \pnR_{U_\ell,\ell}  }
        =
        a_\ell \, \Exp{ U_{\ell} }
\end{IEEEeqnarray}
\end{proof}

\begin{corollary}
    \begin{IEEEeqnarray}{rCl}
        \Exp{U_\ell} = n \, \prod_{i=0}^{\ell-1} \left( 1 - \frac{1}{n-\floor*{\frac{i}{S}}} \right) = n \, \prod_{i=0}^{\ell-1} a_i
    \end{IEEEeqnarray}
\end{corollary}

For ease of exposition, we introduce random variables for the number of received unique symbols, $K_\ell \eqdef n - U_\ell$, and rescaled variants with normalized time and amount of data,
\begin{IEEEeqnarray}{rCl}
    U_\tau^{(n)} \eqdef \frac{1}{n} U_{n \tau},
    \qquad{}
    K_\tau^{(n)} \eqdef \frac{1}{n} K_{n \tau},
    \qquad{}
    \tau \in [0, S].   \IEEEeqnarraynumspace
\end{IEEEeqnarray}

The following theorem provides an analytic expression for the system behavior in the large block length regime ($n\to\infty$):
\begin{theorem}
\label{thm:asymptotic-behavior}
For $n\to\infty$, the fraction of unseen symbols at any given point in time $\tau$ concentrates to its mean and follows
\begin{IEEEeqnarray}{rCl}
    u_S(\tau) \eqdef \lim_{n\to\infty} \Exp{U_\tau^{(n)}} = \left( 1 - \frac{\tau}{S} \right)^S.
\end{IEEEeqnarray}
\end{theorem}
\begin{proof}
First, we show the limit of the expectation.
\begin{IEEEeqnarray}{rCl}
    \IEEEeqnarraymulticol{3}{l}{
        \lim_{n\to\infty} \log\left( \Exp{U_\tau^{(n)}} \right)
        =
        \lim_{n\to\infty} \sum_{\ell=0}^{n\tau-1} \log\left( 1 - \frac{1}{n-\floor*{\frac{\ell}{S}}} \right)
    }\IEEEeqnarraynumspace\\
    &\overset{\text{(a)}}{\approx}&
        \lim_{n\to\infty} \sum_{\ell=0}^{n\tau-1} \log\left( 1 - \frac{1}{n-\frac{\ell}{S}} \right)
    \overset{\text{(b)}}{\approx}
        \lim_{n\to\infty} - \sum_{\ell=0}^{n\tau-1} \frac{1}{n-\frac{\ell}{S}}
        \\
    &=&
        \lim_{n\to\infty} - \sum_{\ell=0}^{n\tau-1} \frac{1}{n} \frac{1}{1-\frac{\ell}{Sn}}
    \overset{\text{(c)}}{\approx}
        \int_0^\tau \frac{1}{\frac{x}{S}-1} \dd x
        \\
    &=&
        S \log\left( 1 - \frac{\tau}{S} \right)
\end{IEEEeqnarray}
Here,
(a) uses $\floor*{\frac{\ell}{S}} \approx \frac{\ell}{S}$ for large $\ell$,
(b) uses $\log(1-x) \approx -x$ for $x \approx 0$,
and (c) uses the convergence of the left Riemann sum to the Riemann integral,
\begin{IEEEeqnarray}{rCl}
    \lim_{n\to\infty}
        \sum_{i=0}^{\alpha n - 1} \frac{1}{n} \, f\!\left( \frac{i}{n} \right)
    =
    \int_{0}^{\alpha} f(x) \dd x,
\end{IEEEeqnarray}
for $\alpha \eqdef \tau$,
$f(x) \eqdef \frac{1}{\frac{x}{S} - 1}$,
and $S > \tau > 0$.

Second, we show the concentration of $U_\tau^{(n)}$ to its expectation as $n\to\infty$.
Note that
\begin{IEEEeqnarray}{C}
    \label{eq:bounds-on-means}
    \forall S, n:
    \quad{}
    \Exp{U_\ell} \leq n \left( 1 - \frac{\ell}{Sn} \right)^S \leq n \exp\left( -\frac{\ell}{n} \right),
    \IEEEeqnarraynumspace\\
    \label{eq:bounds-on-a}
    \forall \ell \;\; \forall \ell' \leq \ell:
    \quad{}
    0 \leq a_\ell \leq a_{\ell'} \leq 1,
    \qquad{}
    2 a_\ell - 1 \leq a_\ell^2.
    \IEEEeqnarraynumspace
\end{IEEEeqnarray}
Hence,
\begin{IEEEeqnarray}{rCl}
    \Var{U_{\ell+1}}
        &\overset{\text{(a)}}{=}& (2 a_\ell - 1) \Exp{U_\ell^2}   \nonumber\\
        &&\quad {}-{} a_\ell^2 \Exp{U_\ell}^2 + (1 - a_\ell) \Exp{U_\ell}   \IEEEeqnarraynumspace\\
        \label{eq:variance-bound}
        &\overset{\text{(b)}}{\leq}& (1 - a_\ell) \Exp{U_\ell} + a_\ell^2 \Var{U_\ell}   \IEEEeqnarraynumspace\\
        &\overset{\text{(c)}}{\leq}& \sum_{i=0}^\ell (1 - a_{\ell-i}) \Exp{U_{\ell-i}} \prod_{j=0}^{i-1} a_{\ell-j}^2   \IEEEeqnarraynumspace\\
        &\overset{\text{(d)}}{\leq}& n \sum_{i=0}^\ell (1 - a_{\ell-i}) \exp\left( - \frac{\ell-i}{n} \right)   \IEEEeqnarraynumspace\\
        &\overset{\text{(e)}}{\leq}& \frac{n}{n - \floor*{\frac{\ell}{S}}} \sum_{i=0}^\ell \exp\left( - \frac{i}{n} \right) = \CO(n).   \IEEEeqnarraynumspace
\end{IEEEeqnarray}
Here,
(a) uses Lemma~\ref{thm:recursive-expected-value} and $\Var{X} = \Exp{X^2} - \Exp{X}^2$,
(b) uses \eqref{eq:bounds-on-a},
(c) results from repeated application of inequality \eqref{eq:variance-bound},
(d) uses \eqref{eq:bounds-on-means} and \eqref{eq:bounds-on-a},
and (e) uses \eqref{eq:bounds-on-a}.
As a result,
using Chebyshev's inequality, for any fixed $\varepsilon > 0$,
\begin{IEEEeqnarray}{C}
    \lim_{n\to\infty} \Prob{ \abs*{ \, U_\tau^{(n)} - \Exp{U_\tau^{(n)}} } \geq \varepsilon }
    \leq
    \lim_{n\to\infty}
    \frac{\Var{U_{n\tau}}}{\varepsilon^2 n^2}
    =
    0.   \IEEEeqnarraynumspace
\end{IEEEeqnarray}
\end{proof}

\makeatletter \newcommand{\pgfplotsdrawaxis}{\pgfplots@draw@axis} \makeatother

\pgfplotsset{axis line on top/.style={
  axis line style=transparent,
  ticklabel style=transparent,
  tick style=transparent,
  axis on top=false,
  after end axis/.append code={
    \pgfplotsset{axis line style=opaque,
      ticklabel style=opaque,
      tick style=opaque,
      grid=none}
    \pgfplotsdrawaxis}
  }
}

\begin{figure}[!t]
    \centering
    \begin{tikzpicture}
        \begin{axis}[
            enlargelimits=false,
            colormap/cool,
            colorbar horizontal,
            colorbar style={
                xlabel={Probability},
                axis line on top,
            },
            xlabel={Transmitted Symbols $\ell$},
            ylabel={Received Unique Symbols $K_{\ell}$},
            legend style={font=\footnotesize},
            legend pos=south east,
            legend cell align={left},
            axis line on top,
            ]
            \input{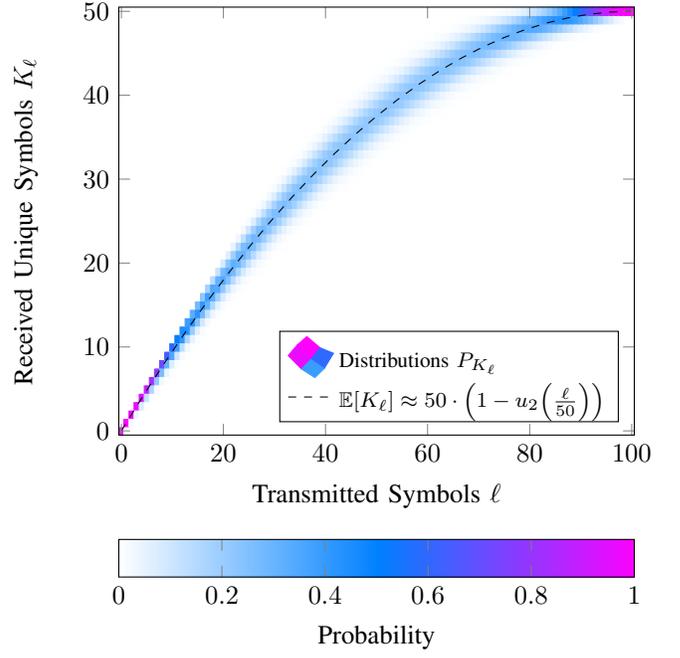}
            \addlegendentry{Distributions $P_{K_{\ell}}$}
            \addplot [dashed,domain=0:100] {50 * (1 - (1 - x/50/2)^2)};
            \addlegendentry{$\Exp{K_\ell} \approx 50 \cdot \left(1 - u_2\!\left(\frac{\ell}{50}\right)\right)$}
        \end{axis}
    \end{tikzpicture}
    \caption{
        PMFs $P_{K_{\ell}}(k_{\ell})$
        of the number of received unique symbols $K_\ell$
        after transmitting $\ell$ symbols,
        for $S=2$ storage systems and block length $n=50$,
        and approximation $\Exp{K_\ell} \approx n \cdot \left(1 - u_S\!\left(\frac{\ell}{n}\right)\right)$
        according to Theorem~\ref{thm:asymptotic-behavior}.
    }
    \label{fig:example-discrete-distribution}
\end{figure}

Ergo,
as $n\to\infty$,
the fraction of unique received symbols
is fully characterized by
$1 - u_S(\tau) = 1 - \left( 1 - \frac{\tau}{S} \right)^S$.
Fig.~\ref{fig:example-discrete-distribution}
visualizes the
PMFs $P_{K_{\ell}}(k_{\ell})$
of the number of received unique symbols $K_\ell$
after transmitting $\ell$ symbols,
for the example of
$S=2$ and $n=50$.
The approximation
based on Theorem~\ref{thm:asymptotic-behavior},
$\Exp{K_\ell} \approx n \cdot \left(1 - u_S\!\left(\frac{\ell}{n}\right)\right)$,
fits well.

To analyze the system for finite block lengths, we turn to the random variable $L$ denoting the number of transmissions required to complete the download (\ie, to receive $k$ unique symbols).
Exactly $\ell$ transmissions are required if the user receives $k-1$ unique symbols in the first $\ell-1$ transmissions and a new symbol in the $\ell$-th transmission,
thus,
\begin{IEEEeqnarray}{rCl}
    P_L(\ell) = P_{U_{\ell-1}}( n-k+1 ) \cdot \pR_{n-k+1,\ell-1}.
\end{IEEEeqnarray}

The average number of transmissions required to complete the download is then obtained as the expected value $\tilde{L} \eqdef \Exp{L}$. We normalize this to $\tilde{\tau} \eqdef \frac{\tilde{L}}{n}$. In a similar fashion for $n\to\infty$, $\tilde{\tau}$ is the fraction of time (\ie, number of transmissions) it takes to complete the download, which can be obtained from
\begin{IEEEeqnarray}{rCl}
    R \eqshould 1 - \left( 1 - \frac{\tilde{\tau}}{S} \right)^S
    ,
    \qquad{}
    \tilde{\tau} = S \ \left( 1 - \sqrt[S]{ 1 - R } \right)
    .
    \IEEEeqnarraynumspace
\end{IEEEeqnarray}

Assume the storage systems use an MDS code of rate $R = \frac{k}{n}$.
We call $\sigma \eqdef \frac{1}{R} = \frac{n}{k}$ the \emph{storage factor} quantifying the storage overhead of the systems due to coding.
Similarly, we call $\delta \eqdef \frac{\tilde{\tau}}{R}$ the \emph{transmission factor},
which quantifies the overhead in transmissions required for uncoordinated vs. coordinated (\ie, where duplicate symbols are avoided) download.

\begin{figure}[!t]
    \centering
    \begin{tikzpicture}
        \begin{axis}[
            xlabel={Storage Factor $\sigma$},
            ylabel={Transmission Factor $\delta$},
            xmin=0.95,
            xmax=2.05,
            ymin=0.95,
            ymax=2.05,
            legend style={font=\footnotesize},
            legend cell align={left},
            restrict x to domain=0:3,
            ]
            \input{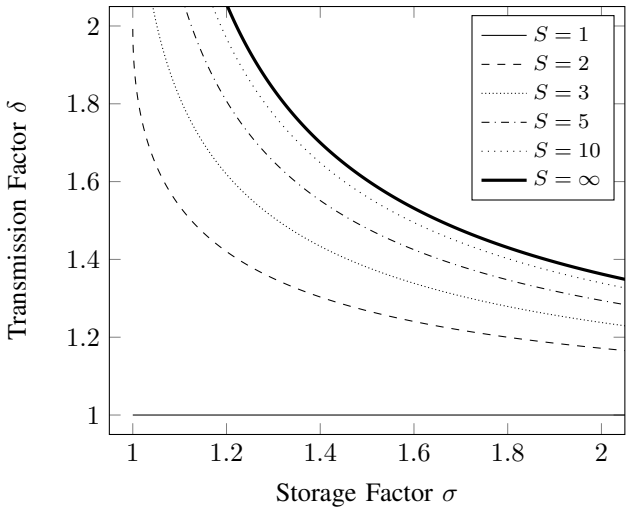}
        \end{axis}
    \end{tikzpicture}
    \caption{Tradeoff between storage factor $\sigma$ and transmission factor $\delta$ for different numbers of storage systems $S$ at large block length $n\to\infty$.}
    \label{fig:storage-transmission-factors}
\end{figure}

\begin{figure}[!t]
    \centering
    \begin{tikzpicture}
        \begin{axis}[
            xlabel={Storage Factor $\sigma$},
            ylabel={Transmission Factor $\delta$},
            xmin=0.95,
            xmax=1.55,
            ymin=1.2,
            ymax=1.55,
            legend style={font=\footnotesize},
            legend cell align={left},
            restrict x to domain=0:3,
            ]
            \input{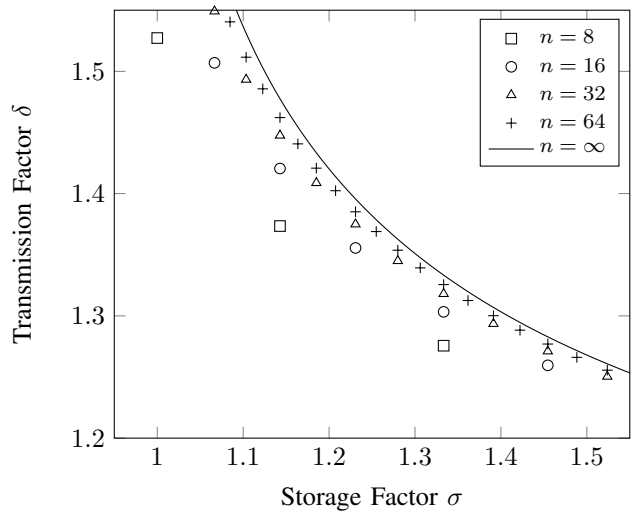}
        \end{axis}
    \end{tikzpicture}
    \caption{
        Tradeoff between storage factor $\sigma$ and transmission factor $\delta$ for number of storage systems $S=2$ and different block lengths $n$.
        Already for reasonably low $n \approx 64$ the system behavior is well predicted using the analytic expression for the asymptotic regime $n\to\infty$.
    }
    \label{fig:storage-transmission-factors-finite}
\end{figure}

Fig.~\ref{fig:storage-transmission-factors} shows the tradeoff between $\sigma$ and $\delta$ for large $n$ and different number of storage systems $S$ all employing the same MDS code.
Clearly, uncoordinated file delivery incurs a loss in transmission efficiency.
The loss gets worse the more independent servers participate in the delivery.
However, even for an arbitrarily large number of sources, the loss
remains bounded
as long as
$\sigma \gneq 1$.
The storage redundancy can be used to mitigate the transmission overhead.

Fig.~\ref{fig:storage-transmission-factors-finite} shows the behavior for some small block lengths $n$ and two servers $S=2$.
In particular for $\sigma$ close to $1$, the system completes downloads slightly faster than for $n\to\infty$.
Note, however, that even for moderate code lengths $n \approx 64$ the asymptotic expression is a good proxy for the system behavior.

\section{Maximum Code Diversity}
\label{sec:different-codes}

After examining the extreme case of minimum code diversity in the previous section, we turn to the opposite extreme, maximum code diversity, in this section.
The possibilities of combining codes of different families have previously been investigated, for instance in \cite{DBLP:conf/IEEEcloud/HellgeM16}.
We explore the inter-operability of three systems employing an RLN-C, RS-C and LDPC-C, respectively, in an experimental case study of the framework depicted in Fig.~\ref{fig:system-model}.
We assume the RLN-C and RS-C are of parameters $[160, 128]_{256}$ and constructed in the usual way \cite{DBLP:journals/tit/HoMKKESL06,reed1960polynomial,DBLP:journals/tcom/BorujenyA16}.
As LDPC-C we use the AR4JA code \cite{DBLP:conf/globecom/DivsalarJDT05,ccsds_greenbook_tmsynccc} of equivalent parameters $[1280, 1024]_2$.

Linear codes can be defined as the nullspace of their parity-check matrices $\BFH$.
To combine the codes, we transform the defining systems of linear equations to the common base field $\IF_2$ as in \cite{Bloemer95anxor-based,DBLP:conf/IEEEcloud/HellgeM16}.
Let $\left\{\BFX\right\}_{y}$ denote the `$y$-th' entry of $\BFX$.
The parity-check equations in $\IF_{p^m}$,
\begin{IEEEeqnarray}{rl}
    \BFH \BFc^\T = \BFzero^\T
    &\iff{}   \nonumber\\
    &\forall i \in \{ 1, ..., n\!-\!k \} :\; \sum_{j=1}^{n} \left\{\BFH\right\}_{ij} \left\{\BFc\right\}_{j} = 0,   \IEEEeqnarraynumspace
\end{IEEEeqnarray}
can be equivalently expressed in $\IF_p$
(albeit with a larger system of equations),
once $+$ and $\cdot$ in $\IF_{p^m}$ are reduced to operations in $\IF_p^m$.
Let $\Ivec: \IF_{p^m}\to\IF_p^m$ be the isomorphism between $\IF_{p^m}$ and $\IF_p^m$, and $\Ivec{}^{-1}: \IF_p^m\to\IF_{p^m}$ its inverse.
Clearly,
\begin{IEEEeqnarray}{rCl}
    \forall \alpha, \beta \in \IF_{p^m}, \gamma \in \IF_p :\; \Ivec\left( \alpha + \gamma \cdot \beta \right) = \Ivec\left(\alpha\right) + \gamma \cdot \Ivec\left(\beta\right).   \IEEEeqnarraynumspace
\end{IEEEeqnarray}
Let $f_\alpha(\beta) \eqdef \alpha \cdot \beta$. Since $+$ and $\cdot$ are distributive, $\tilde{f}_{\alpha}\left( \BFb \right) \eqdef \Ivec\left( f_\alpha\left( \Ivec^{-1}\left( \BFb \right) \right) \right)$ is linear, and thus equivalently represented by a matrix $\Imat\left( \alpha \right) \in \IF_p^{m \times m}$ associated with $\alpha$,
\begin{IEEEeqnarray}{rCl}
    \tilde{f}_{\alpha}\left( \BFb \right) = \Imat\left( \alpha \right) \BFb.
\end{IEEEeqnarray}
We thus get the equivalent system of equations in $\IF_p$,
\begin{IEEEeqnarray}{rCl}
    \forall i \in \{ 1, ..., n\!-\!k \} :\; \sum_{j=1}^{n} \Imat\left( \left\{\BFH\right\}_{ij} \right) \Ivec\left( \left\{\BFc\right\}_{j} \right) = \BFzero.   \IEEEeqnarraynumspace
\end{IEEEeqnarray}

The construction applies analogously to the codes' generator matrices $\BFG$. While technically all three codes have parameters $[1280, 1024]_2$ after this `lifting operation' from $\IF_{256}$ to $\IF_{2}$, we do not pick columns from the lifted generator matrices independently. Rather, we adopt block-aligned erasures, where blocks of eight columns are either jointly erased or jointly present, mimicking the erasure of $\IF_{256}$ rather than $\IF_2$ symbols. For the RLN-C and RS-C, this preserves their properties (both codes behave considerably worse under random erasure of $\IF_2$ sub-symbols), while the LDPC-C's performance is unaffected.

\begin{figure}[!t]
    \centering
    \begin{tikzpicture}
        \begin{ternaryaxis}[
            ternary limits relative=false,colormap/hot,
            colorbar horizontal,colorbar sampled line,colorbar style={xlabel={Probability of Successful Decoding}},
            point meta min=0,point meta max=1,
            xlabel={RLN},
            ylabel={RS},
            zlabel={LDPC},
            label style={sloped},
            minor tick num=1,
            xmin=0, xmax=128,
            ymin=0, ymax=128,
            zmin=0, zmax=128,
            xtick={0,16,...,128},
            ytick={0,16,...,128},
            ztick={0,16,...,128},
            mark size=0.15cm,
            major tick length=0.3cm,
        ]
        \input{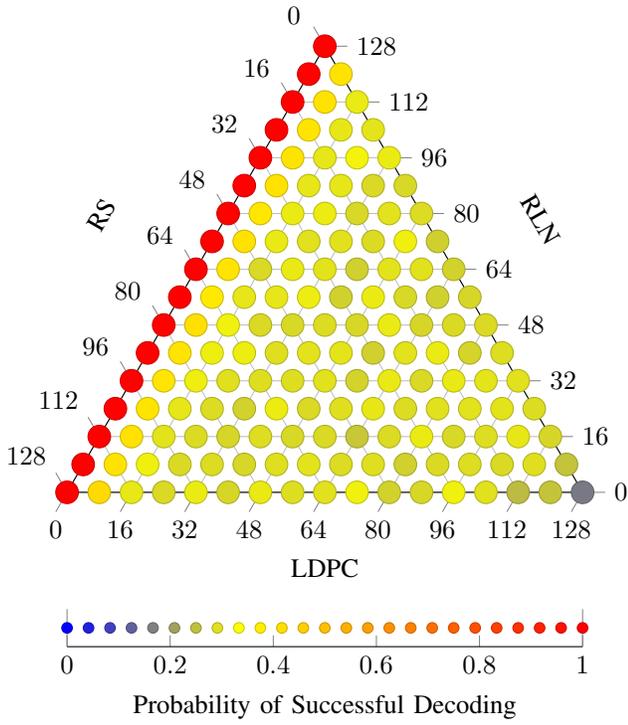}
        \end{ternaryaxis}
    \end{tikzpicture}
    \caption{Probability of successful decoding (mark color) after downloading the minimum number of $k=128$ symbols, consisting of a mix of RLN, RS and LDPC coded symbols (coordinate). The upper center, lower left, and lower right vertex correspond to downloading only RLN, only RS, and only LDPC coded symbols, respectively. Other points correspond to interpolated mixes.}
    \label{fig:combination-different-codes}
\end{figure}

Fig.~\ref{fig:combination-different-codes} shows the probability of the user being able to decode (\ie, the matrix $\BFGtilde$ being full-rank) after downloading the theoretical minimum of $k' = 8 \cdot 128 = 1024$ base field $\IF_2$ symbols in a block-aligned fashion to mimic $\IF_{256}$ symbols.
It can be seen that RS-C and RLN-C symbols can be mixed without performance penalty and the close-to-optimal performance of RLN-Cs over large fields is preserved.
As soon as LDPC-C symbols enter the mix, downloading the theoretical minimum of symbols is likely not sufficient.
However, after downloading one, two and three additional $\IF_{256}$ symbols, the probability of successful decoding reaches $90\%$, $99\%$ and $99.9\%$, respectively,
independent of the mixture
-- which is $2.34\%$ transmission overhead (for $131$ instead of $128$ symbol transmissions).
This is the performance one would expect for an RLN-C over small fields.
In fact, when replacing the LDPC-C with a binary RLN-C, the qualitative system behavior remains unchanged.

Note that erasure decoding using Gaussian elimination over $\BFGtilde$ was performed here.
Its cubic runtime complexity
is acceptable in the case at hand,
as the resulting systems of linear equations are of manageable size,
and the cost of Gaussian elimination
is amortized over the length of the packets.
Belief propagation decoding \cite{DBLP:books/daglib/0027802} is not feasible as neither RLN-Cs nor RS-Cs have the required sparse structure.

As a result of our empirical case study, we conclude that for system design purposes in uncoordinated content delivery, binary LDPC-Cs can be treated as binary RLN-Cs.
The system incurs a small degradation in transmission efficiency due to the small field size, but the overall performance stays close to the theoretical optimum.
Decoding is possible with high probability after downloading only $k + \Delta k$ symbols, with a small $\Delta k$ (in the order of three) number of excess symbols.

In contrast to the degradation observed in Section~\ref{sec:same-codes} from using the same code across multiple storage systems, code diversity enables close-to-optimal uncoordinated content delivery. Note, that lifting to the base field and random recoding transforms any code into an RLN-C over $\IF_2$, providing an effective means to increase code diversity in existing systems.

\section{Conclusion}
\label{sec:conclusion}

In this paper we analyzed the performance of uncoordinated file delivery from multiple coded storage systems.
System implementers and network operators should draw the following conclusions and actionable recommendations from our results.
In Section~\ref{sec:same-codes}, we showed how uncoordinated content delivery from multiple storage systems using an identical code incurs a loss in transmission efficiency
(in comparison to coordinated content delivery)
due to a coupon collector effect.
In Section~\ref{sec:different-codes}, we showed that this degradation can be resolved (almost perfectly) by using different linearly independent codes in all the storage systems.

To leverage the benefits of coding for efficient uncoordinated content delivery from multiple sources, implementers should make the raw storage-coded data and details about the employed coding techniques available to the transport layer, rather than treating storage-coding as a hidden internal of the storage system.
Implementers and operators should employ code families that allow the design of larger numbers of linearly independent codes.
RLN-Cs and RS-Cs over larger field sizes are natural candidates.
Network operators should assign codes to storage systems carefully to maintain proper code diversity (analogous to frequency reuse in mobile network base stations).
When the same code is used multiple times, random remixing/recoding of the stored symbols before transmission can be used to emulate (close-to-optimal) RLN-C performance.
Finally, the procedure outlined in this paper lets system operators change the employed codes without requiring downtime or instantaneous recoding of the whole database.

\IEEEtriggeratref{15}
\bibliographystyle{IEEEtran}
\bibliography{IEEEabrv,IEEEabrv_OWN,references}

\begin{thebibliography}{10}
\providecommand{\url}[1]{#1}
\csname url@samestyle\endcsname
\providecommand{\newblock}{\relax}
\providecommand{\bibinfo}[2]{#2}
\providecommand{\BIBentrySTDinterwordspacing}{\spaceskip=0pt\relax}
\providecommand{\BIBentryALTinterwordstretchfactor}{4}
\providecommand{\BIBentryALTinterwordspacing}{\spaceskip=\fontdimen2\font plus
\BIBentryALTinterwordstretchfactor\fontdimen3\font minus
  \fontdimen4\font\relax}
\providecommand{\BIBforeignlanguage}[2]{{%
\expandafter\ifx\csname l@#1\endcsname\relax
\typeout{** WARNING: IEEEtran.bst: No hyphenation pattern has been}%
\typeout{** loaded for the language `#1'. Using the pattern for}%
\typeout{** the default language instead.}%
\else
\language=\csname l@#1\endcsname
\fi
#2}}
\providecommand{\BIBdecl}{\relax}
\BIBdecl

\bibitem{book_intro_computer_networks}
\BIBentryALTinterwordspacing
P.~L. Dordal, \emph{An Introduction to Computer Networks}, 2018, release
  1.9.16. [Online]. Available: \url{http://intronetworks.cs.luc.edu/}
\BIBentrySTDinterwordspacing

\bibitem{DBLP:conf/sigcomm/KoponenCCEKSS07}
T.~Koponen \emph{et~al.}, ``A data-oriented (and beyond) network
  architecture,'' in \emph{Proc. {ACM} {SIGCOMM} Conf. Appl., Technol.,
  Architectures, and Protocols for Comput. Commun.}, 2007, pp. 181--192.

\bibitem{DBLP:books/daglib/0027592}
M.~Bossert, \emph{Channel Coding for Telecommunications}.\hskip 1em plus 0.5em
  minus 0.4em\relax Wiley, 1999.

\bibitem{DBLP:conf/adl/GoldbergY98}
A.~V. Goldberg and P.~N. Yianilos, ``Towards an archival intermemory,'' in
  \emph{Proc. {IEEE} Forum Res. and Technol. Advances in Digit. Libraries
  ({ADL})}, 1998, pp. 147--156.

\bibitem{DBLP:conf/asplos/KubiatowiczBCCEGGRWWWZ00}
J.~Kubiatowicz \emph{et~al.}, ``{OceanStore}: An architecture for global-scale
  persistent storage,'' in \emph{Proc. Int. Conf. Architectural Support for
  Program. Lang. and Operating Syst. ({ASPLOS})}, 2000, pp. 190--201.

\bibitem{DBLP:conf/hotstorage/RashmiSGKBR13}
K.~V. Rashmi \emph{et~al.}, ``A solution to the network challenges of data
  recovery in erasure-coded distributed storage systems: {A} study on the
  {Facebook} warehouse cluster,'' in \emph{Proc. {USENIX} Workshop Hot Topics
  in Storage and File Syst. ({HotStorage})}, 2013.

\bibitem{acedanski2005good}
S.~Acedanski \emph{et~al.}, ``How good is random linear coding based
  distributed networked storage?'' in \emph{Workshop Network Coding, Theory and
  Applicat. (NETCOD)}, Riva del Garda, Italy, 2005.

\bibitem{pt:03:ldpc}
J.~S. Plank and M.~G. Thomason, ``On the practical use of {LDPC} erasure codes
  for distributed storage applications,'' University of Tennessee, Tech. Rep.
  CS-03-510, September 2003.

\bibitem{DBLP:journals/tcom/ParkLM18}
H.~Park, D.~Lee, and J.~Moon, ``{LDPC} code design for distributed storage:
  Balancing repair bandwidth, reliability, and storage overhead,'' \emph{{IEEE}
  Trans. Commun.}, vol.~66, no.~2, pp. 507--520, 2018.

\bibitem{DBLP:journals/corr/LubyPRMA17}
M.~G. Luby \emph{et~al.}, ``Liquid cloud storage,'' \emph{arXiv:1705.07983v1},
  2017.

\bibitem{DBLP:journals/corr/Luby16}
M.~Luby, ``Capacity bounds for distributed storage,''
  \emph{arXiv:1610.03541v5}, 2018.

\bibitem{DBLP:journals/pvldb/SathiamoorthyAPDVCB13}
M.~Sathiamoorthy \emph{et~al.}, ``{XOR}ing elephants: Novel erasure codes for
  big data,'' \emph{Proc. {VLDB} Endowment}, vol.~6, no.~5, pp. 325--336, 2013.

\bibitem{DBLP:journals/tit/PapailiopoulosD14}
D.~S. Papailiopoulos and A.~G. Dimakis, ``Locally repairable codes,''
  \emph{{IEEE} Trans. Inf. Theory}, vol.~60, no.~10, pp. 5843--5855, 2014.

\bibitem{DBLP:conf/usenix/HuangSXOCG0Y12}
C.~Huang \emph{et~al.}, ``Erasure coding in {W}indows {A}zure storage,'' in
  \emph{Proc. {USENIX} Annu. Tech. Conf.}, 2012, pp. 15--26.

\bibitem{DBLP:journals/pieee/DimakisRWS11}
A.~G. Dimakis \emph{et~al.}, ``A survey on network codes for distributed
  storage,'' \emph{Proc. {IEEE}}, vol.~99, no.~3, pp. 476--489, 2011.

\bibitem{DBLP:journals/ftcit/ShokrollahiL09}
M.~A. Shokrollahi and M.~Luby, ``Raptor codes,'' \emph{Foundations and Trends
  in Commun. and Inf. Theory}, vol.~6, no. 3-4, pp. 213--322, 2009.

\bibitem{DBLP:journals/tos/LiL14}
M.~Li and P.~P.~C. Lee, ``{STAIR} codes: {A} general family of erasure codes
  for tolerating device and sector failures,'' \emph{{ACM} Trans. Storage
  (TOS)}, vol.~10, no.~4, pp. 14:1--14:30, Oct. 2014.

\bibitem{staircasecodes}
M.~Cunche and V.~Roca, ``Optimizing the error recovery capabilities of
  {LDPC}-staircase codes featuring a {Gaussian} elimination decoding scheme,''
  in \emph{Proc. {IEEE} Int. Workshop Signal Process. for Space Commun.
  ({SPSC})}, 2008, pp. 1--7.

\bibitem{DBLP:journals/icl/LivaPC10}
G.~Liva, E.~Paolini, and M.~Chiani, ``Performance versus overhead for fountain
  codes over {$F_q$},'' \emph{{IEEE} Commun. Lett.}, vol.~14, no.~2, pp.
  178--180, 2010.

\bibitem{DBLP:journals/tit/Maddah-AliN19}
M.~A. Maddah{-}Ali and U.~Niesen, ``Cache-aided interference channels,''
  \emph{{IEEE} Trans. Inf. Theory}, vol.~65, no.~3, pp. 1714--1724, 2019.

\bibitem{DBLP:conf/isit/ReisizadehMM18}
H.~Reisizadeh, M.~A. Maddah{-}Ali, and S.~Mohajer, ``Erasure coding for
  decentralized coded caching,'' in \emph{Proc. {IEEE} Int. Symp. Inf. Theory
  ({ISIT})}, 2018, pp. 1715--1719.

\bibitem{DBLP:journals/tcom/PiemonteseA19}
A.~Piemontese and A.~{Graell i Amat}, ``{MDS}-coded distributed caching for low
  delay wireless content delivery,'' \emph{{IEEE} Trans. Commun.}, vol.~67,
  no.~2, pp. 1600--1612, 2019.

\bibitem{DBLP:journals/tit/DebMC06}
S.~Deb, M.~M{\'{e}}dard, and C.~Choute, ``Algebraic gossip: {A} network coding
  approach to optimal multiple rumor mongering,'' \emph{{IEEE} Trans. Inf.
  Theory}, vol.~52, no.~6, pp. 2486--2507, 2006.

\bibitem{DBLP:journals/jacm/Haeupler16}
B.~Haeupler, ``Analyzing network coding (gossip) made easy,'' \emph{J. {ACM}},
  vol.~63, no.~3, pp. 26:1--26:22, 2016.

\bibitem{DBLP:conf/sigcomm/ByersLMR98}
J.~W. Byers \emph{et~al.}, ``A digital fountain approach to reliable
  distribution of bulk data,'' in \emph{Proc. {ACM} {SIGCOMM} Conf. Appl.,
  Technol., Architectures, and Protocols for Comput. Commun.}, 1998, pp.
  56--67.

\bibitem{DBLP:conf/IEEEcloud/HellgeM16}
C.~Hellge and M.~M{\'{e}}dard, ``Multi-code distributed storage,'' in
  \emph{Proc. {IEEE} Int. Conf. Cloud Comput. ({CLOUD})}, 2016, pp. 839--842.

\bibitem{DBLP:journals/tit/HoMKKESL06}
T.~Ho \emph{et~al.}, ``A random linear network coding approach to multicast,''
  \emph{{IEEE} Trans. Inf. Theory}, vol.~52, no.~10, pp. 4413--4430, 2006.

\bibitem{reed1960polynomial}
I.~S. Reed and G.~Solomon, ``Polynomial codes over certain finite fields,''
  \emph{J. {SIAM}}, vol.~8, no.~2, pp. 300--304, 1960.

\bibitem{DBLP:journals/tcom/BorujenyA16}
R.~R. Borujeny and M.~Ardakani, ``A new class of rateless codes based on
  {Reed-Solomon} codes,'' \emph{{IEEE} Trans. Commun.}, vol.~64, no.~1, pp.
  49--58, 2016.

\bibitem{DBLP:conf/globecom/DivsalarJDT05}
D.~Divsalar \emph{et~al.}, ``Protograph based {LDPC} codes with minimum
  distance linearly growing with block size,'' in \emph{Proc. {IEEE} Global
  Commun. Conf. ({GLOBECOM})}, 2005, pp. 1152--1156.

\bibitem{ccsds_greenbook_tmsynccc}
``{TM} synchronization and channel coding -- summary of concept and
  rationale,'' CCSDS SLS-C\&S Working Group, Tech. Rep. 130.1-G-2, November
  2012.

\bibitem{Bloemer95anxor-based}
J.~Blömer \emph{et~al.}, ``An {XOR}-based erasure-resilient coding scheme,''
  Int. Comput. Sci. Inst. (ICSI) Berkeley, Tech. Rep. TR-95-048, August 1995.

\bibitem{DBLP:books/daglib/0027802}
T.~J. Richardson and R.~L. Urbanke, \emph{Modern Coding Theory}.\hskip 1em plus
  0.5em minus 0.4em\relax Cambridge University Press, 2008.

\end{thebibliography}

\end{document}